\documentclass{robtrans} 
\pdfoutput=1
\usepackage{proof}
\usepackage{amssymb}
\usepackage{amsmath}
\usepackage{stmaryrd}

\usepackage{tikz}

\input pdfcolor.tex 
\newtheorem{theorem}{Theorem}

\newcommand\notBox{\mbox{$\not\!\!\!\;\Box$}}
\newcommand\notDiamond{\mbox{$\not\!\!\!\,\,\Diamond$}}

\newcommand{\ndArr}{\vdash}

\DeclareMathOperator{\ndcpls}{\mbox{$\hspace{1.1pt}\vdash\!\!\!\!\!^\ast\hspace{2.1pt}$}}
\DeclareMathOperator{\ndcpl}{\mbox{$\hspace{1.1pt}\vdash\hspace{1.1pt}$}}
\DeclareMathOperator{\hilcplo}{\mbox{$\Vdash\hspace{-7.6pt}^\circ\hspace{1.3pt}$}}
\DeclareMathOperator{\hilcpls}{\mbox{$\Vdash\hspace{-7.6pt}^\ast\hspace{1.3pt}$}}
\DeclareMathOperator{\hilcpl}{\mbox{$\Vdash\hspace{1.1pt}$}}
\DeclareMathOperator{\seqcpls}{\mbox{$\hspace{1.1pt}\Rightarrow\hspace{-10.4pt}^\ast\hspace{4.1pt}$}}
\DeclareMathOperator{\seqcpl}{\mbox{$\hspace{1.1pt}\Rightarrow\hspace{1.1pt}$}}
\newcommand{\stable}{\mathit{stable}}

\newcommand{\rfoc}[2]{#1 \seqcpls \llbracket #2 \rrbracket}
\newcommand{\inv}[2]{#1 \seqcpls #2}
\newcommand{\lfoc}[3]{#1 \seqcpls #2 \gg #3}

\title{Constructive Provability Logic} 
\author{Robert J. Simmons \and Bernardo Toninho}

\begin{abstract} 
  We present {\it constructive provability logic}, an
  intuitionstic modal logic that validates the L\"ob rule of
  G\"odel and L\"ob's provability logic by permitting 
  logical reflection over provability. 
  Two distinct variants of this logic, {\bf CPL}
  and {\bf CPL*}, are presented in natural deduction and
  sequent calculus forms which are then shown to be equivalent.
  In addition, we discuss the
  use of constructive provability logic to justify stratified negation
  in logic programming within an intuitionstic and structural proof theory.

  All theorems presented in this paper are formalized in the Agda proof
  assistant.
  An earlier version of this work was presented at IMLA 2011
  [Simmons and Toninho 2011].
\end{abstract}

\category{F.4.1}{Theory of Computation}{Mathematical Logic}[Proof theory]
            
\terms{Algorithms, Design, Theory, Verification} 
            
\keywords{}

\begin{document}

\maketitle

\noindent
Consider the following propositions (where ``$\supset$'' represents
implication):
\begin{align*}
\forall x.\,\forall y.\,\mathsf{edge}(x,y) & \supset \mathsf{edge}(y,x)\\
\forall x.\,\forall y.\,\mathsf{edge}(x,y) & \supset \mathsf{path}(x,y)\\
\forall x.\,\forall y.\,\forall z.\,\mathsf{edge}(x,y) \supset \mathsf{path}(y,z) & \supset \mathsf{path}(x,z)
\end{align*}

One way to think of these propositions is as rules in a 
{\it bottom-up logic program}. This gives them an operational meaning:
given some known set of facts, a bottom-up logic program uses 
rules to derive more facts.
If we start with the single fact $\mathsf{edge(a,b)}$, we can derive 
$\mathsf{edge(b,a)}$ by using the first rule 
(taking $x = \mathsf{a}$ and $y = \mathsf{b}$), and then, using this new fact,
we can derive $\mathsf{path(b,a)}$ by using the second rule
(taking $x = \mathsf{b}$ and $y = \mathsf{a}$). Finally, from the original
$\mathsf{edge(a,b)}$ fact and the new $\mathsf{path(b,a)}$ fact, we can derive
$\mathsf{path(a,a)}$ using the third rule 
(taking $x = \mathsf{a}$, $y = \mathsf{b}$, and $z = \mathsf{a}$).
Once the only new facts we can derive are facts we already know, we say we 
have reached {\it saturation} --- this will happen in our example when
we have derived $\mathsf{edge(a,b)}$, $\mathsf{edge(b,a)}$, 
$\mathsf{path(a,b)}$, $\mathsf{path(b,a)}$, $\mathsf{path(a,a)}$, and 
$\mathsf{path(b,b)}$. Bottom-up logic programming 
is a very simple and intuitive kind of reasoning,
and it has also 
shown to be an elegant and powerful way of declaratively specifying 
and efficiently solving 
many computational problems, especially in the field of program analysis
(see \cite{whaley05using} for a number of references).

Next, consider the following proposition:
\begin{align*}
\forall x.\,\forall y.\,\mathsf{path}(x,y) \supset \neg \mathsf{edge}(x,y) & \supset \mathsf{noedge}(x,y)
\end{align*}
Intuition says that this is a meaningful statement. 
In our example above, we can derive
$\mathsf{path(a,a)}$, but we can't possibly derive $\mathsf{edge(a,a)}$,
so we should be able to conclude $\mathsf{noedge(a,a)}$. A bottom-up
logic programming semantics based on {\it stratified negation} verifies this 
intuition \cite{przymusinski88declarative}. 
In a stratified logic program made up of the 
four previous rules, we can derive all the 
consequences of the first three rules until saturation is reached. At this 
point, we know everything there is to know 
about facts of the form $\mathsf{edge}(X,Y)$ and $\mathsf{path}(X,Y)$. 
When considering the 
negated premise $\neg \mathsf{edge}(x,y)$ in the fourth rule, we simply check 
the saturated database and conclude that the premise holds if the fact does 
not appear in the database.

Stratified negation would, however, disallow the addition of the following 
rule as paradoxical or contradictory: 
\begin{align*}
\forall x.\,\forall y.\,\mathsf{path}(x,y) \supset \neg \mathsf{edge}(x,y) & \supset \mathsf{edge}(x,y)
\end{align*}
Why is this rule problematic? Operationally, the procedure we used 
for stratified negation no longer really makes sense: we reach
saturation, then conclude that there was no way to prove $\mathsf{edge(a,a)}$,
then use that conclusion to prove $\mathsf{edge(a,a)}$. But we had just
concluded that it wasn't provable! Stratified negation
ensures that we never use the fact that
there is no proof of $A$ to come up with a proof of $A$, either 
directly or indirectly. However, 
stratified negation is an odd property: the program consisting of the single
rule $\neg\mathsf{prop1} \supset \mathsf{prop2}$ is stratified (we consider 
$\mathsf{prop1}$ first, and then we consider $\mathsf{prop2}$), and the program
consisting of the single rule $\neg\mathsf{prop2} \supset \mathsf{prop1}$ is 
also stratified (we consider $\mathsf{prop2}$ first, and then we consider 
$\mathsf{prop1}$), but the two rules cannot be combined as a single stratified
logic program. 

In part due to this non-compositional nature, stratified negation
in logic programming has thus far eluded a treatment by the tools of
structural proof theory. Instead, justifications of negation
in logic programming have universally been of a classical nature based on
the assignment of truth values (Boolean, three-valued, or otherwise) 
to atomic propositions. In this paper, we take
a first step towards a structurally proof-theoretic justification of 
stratified negation in which computation is understood as
proof search for {\it uniform} 
(or {\it focused}) proofs \cite{miller91uniform,andreoli92logic}.
The logic that we present has strong ties to {\bf GL}, 
the G\"odel-L\"ob logic of provability 
\cite{verbrugge10provability},\footnote{{\bf GL} is also 
known variously in the literature as {\bf G}, {\bf L}, 
{\bf Pr}, {\bf PrL}, {\bf KW}, and {\bf K4W}.}
and
we therefore call it {\it constructive provability logic}. 
This connection in our intuitionistic setting  
was anticipated by Gabbay \shortcite{gabbay91modal}, who showed
that {\bf GL} was a natural choice for justifying negation in 
a classical, model-theoretic account of logic programming.

\subsection*{Outline}

Logic programming is our primary motivation, but this article
will mostly focus
on constructive provability logic {\it as a logic}. In Section~\ref{intro},
we develop the ideas behind constructive provability logic.
There are two natural variants of constructive
provability logic with different properties. The ``tethered'' variant of
constructive provability logic, {\bf CPL}, is discussed in 
Section~\ref{sec:cpl}.
The ``de-tethered'' variant of constructive provability logic, 
{\bf CPL*}, is discussed in Section~\ref{sec:cpl*}, 
and in Section~\ref{sec:logprog}
we sketch the use of {\bf CPL*} as a logic programming language.
In Section~\ref{sec:axiom} we consider the
relationship between this logic and classical Hilbert-style
presentations of provability logic, and we conclude in
Section~\ref{sec:conc}. 

In the course of this paper we will give both natural deduction and
sequent calculus presentations of {\bf CPL} and {\bf CPL*}, and show 
that, for each logic, the natural deduction and sequent calculus presentations
are equivalent at the level of provability. Natural deduction presentations
are the most typical way of thinking about proofs and their reductions.
Sequent calculus presentations, on
the other hand, are more useful for proving negative statements about the 
logic (i.e. that a certain fact is \emph{not} provable); such statements come
up frequently in the way we use constructive provability logic.

\section{A judgmental reconstruction of provability logic}\label{intro}

In this section we provide a very brief introduction to the
judgmental methodology that informs our 
development of constructive provability logic. Our presentation
is consistent with Pfenning and Davies' judgmental reconstruction of
modal logic \cite{pfenning01judgmental}, 
which in turn follows Martin L\"{o}f's 1983
Siena Lectures \cite{lof96meanings}.

The key concept behind the judgmental methodology is the separation
between propositions (written $A, B$, etc.) and judgments $J$. A
proposition is a syntactic object that is built up from \emph{atomic
  propositions} using propositional connectives such as implication and
conjunction. Judgments are proved through rules of inference. Thus, we can
talk about proving the judgment $A\;\mathit{true}$ or the judgment
$A\;\mathit{false}$. It is not
meaningful to talk about ``proving $A$'' except as a shorthand way of 
talking about proving the judgment $A\;\mathit{true}$.

When proving a particular judgment, one should be able to
reason from hypotheses. To this end, the concept of an
\emph{hypothetical} judgment, written $J_1 , \dots , J_n \ndArr J$,
comes into play. The conventional interpretation of such a hypothetical
judgment is that $J$ has a proof under the assumptions that $J_1$
through $J_n$ also have proofs. However, the meaning of a
hypothetical judgment is not given to us \emph{a priori}. Rather, we 
\emph{define} the meaning of a hypothetical judgment by defining 
(1) a \emph{hypothesis} principle, (2) a \emph{generalized weakening}
principle, and a (3) \emph{substitution} principle. These principles
arise from the understanding of what a given hypothetical judgment
should mean. The hypothesis principle defines how hypothetical assumptions
are used. The generalized weakening principle defines
primitive operations on hypothetical 
assumptions that do not change the meaning of
a judgment (e.g. ``the order in which we write assumptions does not
matter'', ``all assumptions need not be used in a proof'').
Finally, the substitution
principle defines the conditions under which reasoning through lemmas
is justified. 

Plain-vanilla intuitionistic logic is one of the so-called
{\it structural logics}, and as a structural logic its defining principles
are simple and standard:

\subsection*{\it Defining principles of plain-vanilla intuitionistic logic:}
\begin{itemize}
\item {\it Hypothesis principle}: If $A~\mathit{true} \in \Psi$, then $\Psi \ndArr A~\mathit{true}$.
\item {\it Generalized weakening principle}: If $\Psi \subseteq \Psi'$ and $\Psi \ndArr A~\mathit{true}$, then $\Psi' \ndArr A~\mathit{true}$.
\item {\it Substitution principle}: If $\Psi \ndArr A~\mathit{true}$ and $\Psi, A~\mathit{true} \ndArr C~\mathit{true}$, then $\Psi \ndArr C~\mathit{true}$.
\end{itemize}

These principles have an interesting
character. While they are, in some sense, the last thing we need
to consider when defining a logic (i.e. after defining the logic, they
are theorems we need to prove about the system), the judgmental
methodology tells us that these principles are also the {\it first} things
that need to be considered. 
Philosophically, this arises from the fact that these principles flow
from our understanding of the meaning of the hypothetical judgment.
More pragmatically, generalized 
weakening and substitution are necessary as we perform
sanity checks on the rules that define individual connectives.

\subsection{Natural deduction in the judgmental methodology}

The judgmental methodology is generally played out in the setting of
natural deduction. In natural deduction, the meaning of a
logical connective is given by two sets of rules: the
\emph{introduction} rules, stating how we can come to know (that is, prove) 
of the truth of
that connective, and the \emph{elimination} rules, defining
how we can use the knowledge (that is, the proof) 
of that proposition's truth. For
instance, implication $A \supset B$ is defined by one
introduction rule ${\supset}I$ and one elimination rule ${\supset}E$:
$$
\infer[{\supset}I]{\Psi \ndArr A \supset B \;true}
{\Psi , A\;true \ndArr B\;true}
\qquad
\infer[{\supset}E]{\Psi \ndArr B\;true}
{\Psi \ndArr A\supset B\;true & \Psi \ndArr A\;true}
$$
In natural deduction, the 
sanity checks that we perform on a definition like this are called
\emph{local soundness} and \emph{local completeness}. Local soundness
ensures that the introduction rules are strong enough with respect to the
elimination rules, whereas local completeness ensures that the
introduction rules are not too strong with respect to the elimination rules. 

\paragraph*{Local soundness}

Consider a proof $\mathcal D$ of the judgment $\Psi \ndArr C\;true$ where
the last rule is an elimination rule (in the case for implication, the
elimination rule is ${\supset} E$ and so we have two subproofs, one of
$\Psi \ndArr A \supset C\; true$ -- call it $\mathcal D_1$ -- 
and another of $\Psi \ndArr A\;true$ -- call it $\mathcal D_2$). 
Since the rule is an elimination
rule, it is necessarily the case that one of the subproofs mentions
the relevant connective (in the case for implication, the
first sub-proof $\mathcal D_1$
mentions the connective). 
Local soundness is the property that, if the last rule in the 
connective-mentioning premise is an introduction rule, 
then both the introduction
rule and the elimination are unnecessary. To show this,
we build a
proof of $\Psi \vdash C\;true$ using only the premises of the
introduction rule and any other premises of the elimination rule. In
our example with implication, we can obtain this new proof by
appealing to the substitution principle for the subproofs labeled
$\mathcal D_2$ and $\mathcal D_1'$:
\[
\begin{array}{c}
\infer[{\supset}E]
{\Psi \ndArr C~{\it true}\mathstrut}
{\infer[{\supset}I]{\Psi \ndArr A \supset C~{\it true}\mathstrut}
{\deduce{\Psi, u : A~{\it true} \ndArr C~{\it true}\mathstrut}{\mathcal D'_1 \mathstrut}}
 & \deduce{\Psi \ndArr A~{\it true}\mathstrut}{\mathcal D_2 \mathstrut}\mathstrut}
\end{array}
\quad
\Longrightarrow_R
\quad
\begin{array}{c}
\deduce{\Psi \vdash C~{\it true} \mathstrut}{[\mathcal D_2/u] \mathcal D_1' \mathstrut}
\end{array}
\]
Note that, 
following standard conventions, we gave the label
$u$ to the premise $A\;\mathit{true}$ 
in the hypothetical judgment to make it clear what we were 
substituting for.

\paragraph*{Local completeness}
Where local soundness is witnessed by a proof reduction, local
completeness is witnessed by a proof expansion: given an arbitrary
proof of the truth of connective we are interested
in, we show that by applying the
elimination rules and then applying the introduction rules we can
reconstruct the initial proof. In the expansion below, we obtain
$\mathcal D'$ by applying the generalized
weakening principle to the given proof
$\mathcal D$:
\[
\begin{array}{c}
\deduce{\Psi \vdash A \supset B~{\it true} \mathstrut}{\mathcal D \mathstrut}
\end{array}
\quad
\Longrightarrow_E
\quad
\begin{array}{c}
\infer[{\supset}I]{\Psi \vdash A \supset B~{\it true} \mathstrut}
{\infer[{\supset}E]{\Psi, A~{\it true} \vdash B~{\it true} \mathstrut}
 {\deduce{\Psi, A~{\it true} \vdash A \supset B~{\it true} \mathstrut}
  {\mathcal D' \mathstrut}
  & 
  \infer[\it hyp]{\Psi, A~{\it true} \vdash A~{\it true}\mathstrut}{}
}}
\end{array}
\]
We also used the hypothesis principle in the above example: in natural
deduction systems, the hypothesis 
principle always holds trivially due the presence of the 
rule we labeled {\it hyp}.

\subsection{Reflection over an accessibility relation}\label{sec:reflect1}

Having reviewed the judgmental methodology, we will now perform a sort
of warm-up exercise to introduce the idea of definitional reflection in the 
presentation of a logic \cite{schroederheister93rules}. This warm-up logic,
which we name {\bf DML}
(for ``{\bf D}efinitional {\bf M}odal {\bf L}ogic''), is
recognizably similar to {\bf IK}, 
the intuitionistic Kripke semantics for modal logic 
presented by Simpson \shortcite{simpson94proof}.

Kripke semantics for modal logic are characterized by {\it worlds} and an 
{\it accessibility relation} that describes the relationship between worlds.
We will use as a running example an accessibility relation
with three worlds, $\alpha$, $\beta$, and $\gamma$, such that
$\alpha \prec \beta$ (we say ``$\beta$ is accessible from $\alpha$''), 
$\alpha \prec \gamma$, and $\beta \prec \gamma$. 

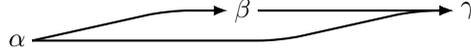
\begin{center}
\begin{tikzpicture} 
\draw  (-5cm,0) node{$\beta$};
\draw (-2cm,0) node{$\gamma$};
\draw (-8cm,-.4cm) node{$\alpha$};
\pgfsetarrowsstart{latex} 
\pgfsetlinewidth{.3pt} 
\pgfusepath{stroke} 
\draw[thick] (-2.2cm,0cm) -- (-4.8cm,0cm); 
\draw[thick,rounded corners=8pt] (-2.2cm,0cm) -- (-2.7cm,0cm) -- (-4.5cm,-.4cm) -- (-7.8cm,-.4cm);
\draw[thick,rounded corners=8pt] (-5.2cm,0cm) -- (-6cm,0cm) -- (-7.8cm,-.4cm);
\pgfsetlinewidth{.3pt} 
\end{tikzpicture} 
\end{center}


The proof theory of {\bf DML} is 
parametrized over an arbitrary accessibility relation; the
three-world accessibility relation above is only one possible example.
The hypothetical judgment for this logic takes the form
$A_1[w_1],\ldots,A_n[w_n] \ndArr C[w]$, where $C$ and
the $A_i$ are propositions and $w$ and the $w_i$
are worlds. {\bf DML} is also a structural logic, so its judgmental
principles are straightforward:

\subsection*{\it Defining principles of {\bf DML}:}
\begin{itemize}
\item {\it Hypothesis principle}: If $A[w] \in \Gamma$, then $\Gamma \ndArr A[w]$.
\item {\it Generalized weakening principle}: If $\Gamma \subseteq \Gamma'$ and $\Gamma \ndArr A[w]$, then $\Gamma' \ndArr A[w]$.
\item {\it Substitution principle}: If $\Gamma \ndArr A[w]$ and $\Gamma, A[w] \ndArr C[w']$, then $\Gamma \ndArr C[w']$.
\end{itemize}

In {\bf DML}, as in Simpson's {\bf IK}, worlds and 
the accessibility relation are critical to the
definition of the modal operators. Consider the definition of modal
possibility, $\Diamond A$. The Kripke interpretation of modal possibility
is that $\Diamond A$ is
true at world $w$ if there exists some accessible world $w'$ where $A$ is true.
The introduction rule for modal possibility directly reflects this 
interpretation:
\[
\infer[\Diamond I]
{\Gamma \ndArr \Diamond A[w] \mathstrut}
{w \prec w' & \Gamma \ndArr A [w'] \mathstrut}
\]
\noindent
The elimination rule for modal possibility is where the use of definitional
reflection becomes important. If we can prove that $\Diamond A$ is
true at the world $w$, we can use case analysis over the 
pre-defined accessibility relation
to look up all the worlds $w'$ such that $w \prec w'$ holds; for each
such $w'$, we must prove the ultimate conclusion using the additional
hypothesis $A[w']$. This is expressed by the following inference rule:
\[\infer[\Diamond E]
{\Gamma \ndArr C [w''] \mathstrut}
{\Gamma \ndArr \Diamond A[w] & 
\forall w'.\, w \prec w' \longrightarrow \Gamma, A [w'] \ndArr C [w''] 
\mathstrut}
\]

In our aforementioned example, there are two worlds $w'$ 
such that $\alpha \prec w'$ holds. Therefore, to eliminate
a proof of $\Diamond A[\alpha]$, 
we must consider the case where $A$ holds at world
$\beta$ and the case where $A$ holds at world $\gamma$. Similarly,
because there are zero worlds $w'$ such that $\gamma \prec w'$ holds,
a proof of $\Diamond A[\gamma]$ is contradictory
and can be used to prove anything at all. These two 
derivable special cases of the possibility elimination
rule can be written as follows:
\[
\infer[\Diamond E_\alpha]
{\Gamma \ndArr C[w''] \mathstrut}
{\Gamma \ndArr \Diamond A[\alpha] \mathstrut
& \Gamma, A[\beta] \ndArr C[w'']
& \Gamma, A[\gamma] \ndArr C[w'']}
\qquad
\infer[\Diamond E_\gamma]
{\Gamma \ndArr C[w''] \mathstrut}
{\Gamma \ndArr \Diamond A[\gamma] \mathstrut}
\]
This elimination rule is what makes {\bf DML} 
strikingly different, and seemingly stronger, than Simpson's {\bf IK}.
In {\bf IK}, it would not be possible to prove
$\cdot \vdash \Diamond A \supset \bot [\gamma]$, but in {\bf DML} this
is a simple use of the ${\supset}I$ and ${\Diamond}E_\gamma$ rules. This 
strength comes at a price, of course. Any reasoning in {\bf IK} is valid
in a larger accessibility relation, but in {\bf DML}, the 
aforementioned hypothetical judgment
$\cdot \vdash \Diamond A \supset \bot [\gamma]$ would no longer be 
valid if the accessibility relation was made larger in certain
ways (for example, by making $\gamma$ accessible from itself). 

It is possible, at least in this simple case, to see $\Diamond E$ as merely
a rule schema that, once given an accessibility
relation, stamps out an appropriate number of rules.
However, as suggested by Zeilberger \shortcite{zeilberger08focusing}, 
it is more auspicious
to take this higher-order formulation of definitional reflection
at face value: the second premise of the $\Diamond E$ rule is actually 
a (meta-level) mapping -- a function -- 
from facts about the accessibility relation to 
derivations. This interpretation becomes relevant when we discuss local
soundness.

To show local soundness, we use functional application to discharge the 
higher-order premises, so that $(\mathcal D_2 \,w'\,\mathcal A_1)$ below is
a derivation of the hypothetical judgment $\Gamma, u : A[w'] \vdash C[w'']$. 
\[
\begin{array}{c}
\infer[{\Diamond}E]
{\Gamma \ndArr C[w'']\mathstrut}
{\infer[{\Diamond}I]{\Gamma \ndArr \Diamond A [w']\mathstrut}
{\deduce{w \prec w' \mathstrut}{\mathcal A_1 \mathstrut}
 & 
 \deduce{\Gamma \ndArr A[w']\mathstrut}{\mathcal D'_1 \mathstrut}}
 & \deduce{\forall{w^*}. w \prec w^* \longrightarrow \Gamma, u : A[w^*] \vdash C[w'']}{\mathcal D_2 \mathstrut \qquad\quad~}\mathstrut}
\end{array}
\qquad\qquad\qquad
\]\vspace{-10pt}\[
\qquad\qquad\qquad\qquad\qquad\qquad\qquad\qquad\qquad\qquad\qquad\qquad
\Longrightarrow_R
\quad
\begin{array}{c}
\deduce{\Gamma \vdash C[w''] \mathstrut}{[\mathcal D'_1/u](\mathcal D_2 \,w'\,\mathcal A_1) \mathstrut}
\end{array}
\]

Local completeness is a bit difficult to write clearly in the traditional 
two-dimensional notation used for proofs. It begins like this:
\[
\begin{array}{c}
\deduce{\Gamma \ndcpl \Diamond A[w] \mathstrut}{\mathcal D \mathstrut}
\end{array}
\quad
\Longrightarrow_E
\quad
\begin{array}{c}
\infer[\Diamond E]{\Gamma \ndcpl \Diamond A[w]}
{\deduce{\Gamma \ndcpl \Diamond A[w]}{\mathcal D}
& 
\deduce{\forall w'. w \prec w' \longrightarrow \Gamma, A[w'] \vdash \Diamond A [w]}{???\qquad~~}}
\end{array}
\]
We discharge the remaining proof obligation marked $???$ above with
a lemma: we must prove that for all $w'$, $w \prec w'$ implies 
$\Gamma, A[w'] \vdash \Diamond A[w]$. If we label the given premise
$w \prec w'$ as $\mathcal A$, this fact is be established 
by the following schematic derivation:
\[
\infer[{\Diamond}I]
{\Gamma, A[w'] \vdash \Diamond A [w] \mathstrut}
{\deduce{w \prec w' }{\mathcal A \mathstrut}
&
\infer[\it hyp]
{\Gamma, A[w'] \vdash A [w'] \mathstrut}
{}
\mathstrut}
\]
This proves our lemma, which in turn
suffices to show local completeness for modal possibility, 
ending our discussion
of the system {\bf DML}.

\subsection{Reflection over provability}\label{sec:reflect2}

The system {\bf DML} was just a warm-up that introduced
reflection over the definition of an accessibility relation. We will now 
introduce constructive provability logic by additionally
using reflection over provability. 
In {\bf DML}, a proof of 
$\Diamond A [w]$ allows us to assume (by the addition
of a new hypothetical assumption) that $A$ is true at one of the worlds
$w'$ accessible from $w$; if there is no such world $w'$, the
assumption is contradictory. In constructive provability logic, 
a proof of $\Diamond A [w]$ will allow us to assume 
that $A$ is {\it provable given the current set of hypotheses} 
at one of the worlds
$w'$ accessible from $w$. If $A$ is not currently provable at some world
$w'$ accessible from $w$, the assumption is contradictory.

As a specific example, if $Q$ is an arbitrary atomic
proposition, $\bot$ is the proposition representing falsehood, 
and we use the accessibility relation from the previous section, then 
in constructive provability logic we can prove 
$\Diamond Q [ \alpha ] \ndArr \bot [ \alpha ]$ by the use of
reflection over logical provability.
It is possible to show, using techniques that we will introduce later,
that there is {\it no} proof 
of $\Diamond Q [ \alpha ] \ndArr Q [ \beta ]$ and 
{\it no} proof of 
$\Diamond Q [ \alpha ] \ndArr Q [ \gamma ]$.
This, in turn, allows us to conclude that asserting
that $Q$ is currently provable at one of the worlds $w'$ accessible
from $\alpha$ is contradictory.
The same judgment 
$\Diamond Q [ \alpha ] \ndArr \bot [ \alpha ]$ would {\it not} have been
provable in {\bf DML}. In order to use a proof of
$\Diamond Q [ \alpha ]$ in {\bf DML}, we would have to prove
both $\Diamond Q [ \alpha ], Q [ \beta ] \ndArr \bot [ \alpha ]$ and
$\Diamond Q [ \alpha ], Q [ \gamma ] \ndArr \bot [ \alpha ]$, and
neither of these hypothetical judgments are, in fact, provable.

\subsubsection{The weakening principle for constructive provability logic}

The discussion above is enough to make it clear 
that the generalized weakening principle from {\bf DML}
will not be acceptable for constructive provability logic. 
In {\bf DML}, the weakening principle
asserts that, if we can prove $\Gamma \vdash \bot[\alpha]$, then we can always
also prove $\Gamma, Q [ \beta ] \vdash \bot [ \alpha ]$. Compare this
to the previous discussion where we counted on there being no proof of
the hypothetical judgment $\Diamond Q [ \alpha ] \ndArr Q [ \beta ]$. 
If we weaken the context with the additional judgment $Q [ \beta ]$, we 
get a hypothetical judgment 
$\Diamond Q [ \alpha ], Q [\beta] \ndArr Q [ \beta ]$ that {\it is}
provable, invalidating our reasoning.

This illustrates that constructive provability logic 
must avoid some forms of weakening. 
To this end, we define a new partial order on contexts
that is indexed by a world $w$,
written as $\Gamma \subseteq_w \Gamma'$. This relation holds exactly when:
\begin{itemize}
\item For all $w'$ such that $w \prec^* w'$, $A [ w' ] \in \Gamma$ implies
$A [ w' ] \in \Gamma'$, and
\item For all $w'$ such that $w \prec^+ w'$, $A [ w' ] \in \Gamma'$ implies
$A [ w' ] \in \Gamma$. 
\end{itemize}
Here, $w \prec^* w'$ is the reflexive and 
transitive closure of the accessibility
relation and $w \prec^+ w'$ is the transitive closure of the accessibility
relation. This indexed subset relation $\subseteq_w$
acts like the normal subset relation when
dealing with judgments $A[w]$, but for assumptions at worlds $A[w']$
where $w'$ is transitively accessible
from $w$, only contraction and exchange are allowed. Assumptions $A[w']$
where $w'$ is neither equal to $w$ nor transitively accessible from $w$ are 
completely unconstrained and can be added or removed without restriction.

With our new partial order, we can present two of the defining principles
of constructive provability logic. 

\subsection*{\it Partial defining principles of constructive provability logic:}
\begin{itemize}
\item {\it Hypothesis principle}: If $A[w] \in \Gamma$, then $\Gamma \ndArr A[w]$.
\item {\it Generalized weakening principle}: If $\Gamma \subseteq_w \Gamma'$ and $\Gamma \ndArr A[w]$, then $\Gamma' \ndArr A[w]$.
\end{itemize}

\noindent
We omit the substitution principle for now, because it is different in the
two different variants of constructive provability logic that we present
in this paper.

\subsubsection{Restrictions on accessibility relations and the form of rules}

Reflection over provability must be done with care. It would be logically
inconsistent to modify our previous elimination
rule for modal possibility by turning the hypothesis $A[w']$ into 
a higher-order assumption $\Gamma \ndArr A[w']$ like this:

\[\infer[\Diamond E_{\it bad}]
{\Gamma \ndArr C [w''] \mathstrut}
{\Gamma \ndArr \Diamond A [w] & 
\forall w'.\, w \prec w' 
\longrightarrow \Gamma \ndArr A [w'] \longrightarrow \Gamma \ndArr C [w''] 
\mathstrut}
\]
This definition can lead to logical inconsistency 
because the hypothetical judgment
$\Gamma \ndArr A[w']$ occurs to the left of an arrow 
in a rule that is ostensibly defining the hypothetical judgment. 
In {\bf DML} this
was no issue: we stipulated that the accessibility relation was 
definable independently
from the hypothetical judgment. 

To make the definition of constructive
provability logic well-formed, we take the position that the hypothetical 
judgment $\Gamma \ndArr A[w]$ is defined one world at a time. If we then
restrict the accessibility relation so that it is {\it converse well-founded}
(irreflexive, no cycles or infinite ascending chains), 
when $w \prec w'$, then we can hope to define
$\Gamma \ndArr A [ w' ]$ before 
$\Gamma \ndArr A [ w ]$ in the same way we defined 
the accessibility relation
$w \prec w'$ before $\Gamma \ndArr A [ w ]$ in {\bf DML}.

If we are trying to define provability one world at a time, the
problem with $\Diamond E_{\it bad}$ is the relationship
(or lack thereof) between 
$\Gamma \ndArr A [ w' ]$, which we are reflecting over, and 
$\Gamma \ndArr C [ w'' ]$, which we are defining. To fix this, we must
ensure that $w'$ is accessible from $w''$ in one or more steps, and therefore
defined before $w''$. There are two obvious ways to do this,
which give rise to the two variants of constructive provability logic,
{\bf CPL} and {\bf CPL*}.

\subsubsection{Tethered constructive provability logic}\label{sec:previewcpl}

 Because $w \prec w'$,
the simplest solution is to force $w$ to be equal to $w''$; this results in
the following ``tethered'' (in the sense that the world in the premise 
$\Diamond A[w]$ is tethered to the conclusion $C[w]$) rule for modal possibility: 
\[\infer[\Diamond E_{\bf CPL}]
{\Gamma \ndArr C [w] \mathstrut}
{\Gamma \ndArr \Diamond A [w]& 
\forall w'.\, w \prec w' 
\longrightarrow \Gamma \ndArr A [w'] 
\longrightarrow \Gamma \ndArr C [w] 
\mathstrut}
\]
We call this tethered version of constructive provability logic {\bf CPL},
and show the rules for modal necessity to be
locally sound and complete in Section~\ref{sec:cpl}.

\subsubsection{De-tethered constructive provability logic}\label{sec:previewcpls}

The tethered proof theory of {\bf CPL}
can be viewed as unnecessarily restrictive. 
To fix the inconsistent left rule $\Diamond E_{\it bad}$, all that is really
necessary according to the discussion above is for 
provability at $w'$ to be defined before provability at $w''$. 
We can ``de-tether'' the logic somewhat by allowing both the
case where
$w$ is the same as $w''$ and the case where $w$ is transitively accessible
from $w''$ (this is achieved by adding a premise $w'' \prec^* w$). This
is sufficient to ensure that $w'$ will be transitively accessible from $w''$
($w'' \prec^+ w'$), ensuring that provability at $w'$ will be defined
before provability at $w''$ as required. The de-tethered
elimination rule for modal possibility in constructive provability
logic looks like this:
\[\infer[\Diamond E_{\bf CPL*}]
{\Gamma \ndcpls C [w''] \mathstrut}
{w'' \prec^* w &
\Gamma \ndcpls \Diamond A [w] & 
\forall w'.\, w \prec w' 
\longrightarrow \Gamma \ndcpls A [w'] \longrightarrow \Gamma \ndcpls C [w''] 
\mathstrut}
\]
We call the de-tethered variant of constructive provability
logic {\bf CPL*}. To distinguish the two similar logics, in the
subsequent discussion we will
write the hypothetical judgment for {\bf CPL} as 
$\Gamma \ndcpl A[w]$ and write the hypothetical judgment for
{\bf CPL*} as $\Gamma \ndcpls A[w]$.

\subsection{A note on formalization}\label{sec:note}

Both variants of 
constructive provability logic and their 
metatheory have been formalized in the
Agda proof assistant, an implementation of the constructive
type theory of Martin L\"of \cite{norell08towards}. This development is 
available from
\url{https://github.com/robsimmons/agda-lib/tree/cpl}.

With two exceptions, all of the results in this paper are fully verified by
Agda. The most significant exception is that Agda cannot verify that rules
such as $\Diamond E_{\bf CPL}$ and
$\Diamond E_{\bf CPL*}$ above avoid logical inconsistency.
This is because Agda's positivity checker, which ensures that data-types
are not self-referential, does not understand the critical relationship
between the logical rules and the 
converse well-founded accessibility relation. The result is that 
the positivity checker must be disabled when we encode the
definitions of {\bf CPL} and {\bf CPL*}. This issue
is discussed further in the technical report along with
potential resolutions \cite{simmons10principles}. One key
point is that any finite accessibility relation can be instantiated
without running afoul of the positivity issue, so we can restrict any concerns 
to instantiations of constructive provability logic with infinite converse
well-founded accessibility relations.

The second issue is that, 
due to the complexity of the
de-tethered cut admissibility proof, 
Agda runs out of memory and crashes 
when attempting to verify that this proof terminates. Therefore,
we must turn off the termination checker when dealing with this proof.
Arguably, this shortcoming is due
to the fact that Agda does not allow the user to specify an induction
metric -- rather, it synthesizes all possible induction metrics
and then checks them. However, we can state an induction metric and verify
by hand that this induction metric is obeyed in the proof.

\section{{\bf CPL}, Tethered constructive provability logic}\label{sec:cpl}

In this section, we will present the defining principles, natural
deduction, and sequent calculus for the tethered variant of 
constructive provability logic, {\bf CPL}. Mirroring the tethered presentation
of rules outlined in 
Section~\ref{sec:previewcpl}, the substitution principle in
{\bf CPL} is tethered: the hypothesis being discharged, $A[w]$, is at the
same world as the consequent $C[w]$.

\subsection*{\it Defining principles of {\bf CPL}:}
\begin{itemize}
\item {\it Hypothesis principle}: If $A[w] \in \Gamma$, then $\Gamma \ndcpl A[w]$.
\item {\it Generalized weakening principle}: If $\Gamma \subseteq_w \Gamma'$ and $\Gamma \ndcpl A[w]$, then $\Gamma' \ndcpl A[w]$.
\item {\it Substitution principle}: If $\Gamma \ndcpl A[w]$ and 
 $\Gamma, A[w] \ndcpl C[w]$, then $\Gamma \ndArr C[w]$.
\end{itemize}

\begin{figure}[tp]
\[
\infer[{\it hyp}]{\Gamma,
  A[w]  \ndcpl A[w] }{}
 \quad
\infer[\bot{E}]{\Gamma \ndcpl C[w]}{\Gamma \ndcpl \bot[w] }
\]
\[
\infer[{\supset}I]{\Gamma \ndcpl A\supset B[w] }{\Gamma, A[w]
  \ndcpl B[w]} \quad 
\infer[{\supset}E]{\Gamma \ndcpl B[w] }{\Gamma \ndcpl A\supset B[w]
  & \Gamma \ndcpl A[w]} 
\]
\[
\infer[{\Diamond}I]{\Gamma \ndcpl \Diamond A[w]}{w\prec w' \quad \Gamma
\ndcpl A[w']} \quad
\infer[{\Box}I]{\Gamma \ndcpl \Box A [w]}{\forall w'.\, w\prec w'
  \longrightarrow \Gamma \ndcpl A[w']  } 
\]
\[
\infer[{\Diamond}E]{\Gamma \ndcpl C[w]}
{\Gamma \ndcpl \Diamond A[w] & \forall
  w'.\, w\prec w' \longrightarrow \Gamma \ndcpl A[w'] \longrightarrow
  \Gamma  \ndcpl C[w]}
\]
\[
\infer[{\Box}{E}]{\Gamma\ndcpl C[w]}{ \Gamma \ndcpl \Box A[w]  & (\forall
  w'.\, w\prec w' \longrightarrow \Gamma\ndcpl A[w']) \longrightarrow 
 \Gamma \ndcpl C[w]}
\]
\caption{Intuitionistic {\bf CPL} natural deduction}
\label{fig:teth_nd}
\end{figure}
The natural deduction rules for {\bf CPL} are presented in Fig.~\ref{fig:teth_nd}.
Implication, atomic propositions and falsehood are defined as per usual in
natural deduction presentations of logic. 
The introduction rule for modal possibility is visually the same as the
rule from {\bf DML}, and the elimination rule was presented in 
Section~\ref{sec:previewcpl}, but we have yet
to show these rules locally sound and complete. Local soundness is 
witnessed by the following reduction; as in the local soundness
proof for possibility in {\bf DML},
the higher-order proof $\mathcal D_3$ is used as
a function -- we apply it to $w'$, $\mathcal A_1$, and $\mathcal D_2$ in
order to obtain the necessary proof:
\[
\begin{array}{c}
\infer[{\Diamond}E]
{\Gamma \ndArr C[w]\mathstrut}
{\infer[{\Diamond}I]{\Gamma \ndcpl \Diamond A[w]\mathstrut}
{\deduce{w\prec w'\mathstrut}{\mathcal A_1 \mathstrut} & \deduce{\Gamma
    \ndcpl A[w']\mathstrut}{\mathcal D_2}}
 & \deduce{\forall w'. w \prec w' \longrightarrow \Gamma \ndcpl A[w']
   \longrightarrow \Gamma \ndcpl C[w]}{\mathcal D_3 \mathstrut}\mathstrut}
\end{array}
\qquad\qquad
\]
\vspace{-10pt}
\[
\qquad\qquad\qquad\qquad\qquad\qquad\qquad\qquad\qquad\qquad\qquad\qquad\qquad
\quad
\Longrightarrow_R
\quad
\begin{array}{c}
\deduce{\Gamma \ndcpl C[w] \mathstrut}{\mathcal D_3\,w'\,\mathcal
  A_1\,\mathcal D_2\mathstrut}
\end{array}
\]
Local completeness also holds for modal possibility, although the
expansion that witnesses the property is somewhat surprising:
\[
\begin{array}{c}
\deduce{\Gamma \ndcpl \Diamond A[w] \mathstrut}{\mathcal D \mathstrut}
\end{array}
\quad
\Longrightarrow_E
\quad
\begin{array}{c}
\infer[\Diamond E]{\Gamma \ndcpl \Diamond A[w]}
{\deduce{\Gamma \ndcpl \Diamond A[w]}{\mathcal D}
& \Diamond I}
\end{array}
\]
We expand a proof of $\Diamond A[w]$ by applying $\Diamond E$ to the
given derivation \emph{and} to the actual rule of $\Diamond I$. The
higher-order premise for $\Diamond E$ for this proof requires us to
prove the following meta-theorem: ``If $w\prec w'$ and $\Gamma \ndcpl
A[w']$ then $\Gamma \ndcpl \Diamond A[w]$.'' This theorem is
immediately true by application of the $\Diamond I$ rule to the assumptions.

All that remains is a discussion of modal necessity.
Whereas modal possibility
has an existential character (there 
{\it exists} some accessible world where $A$ is true), modal necessity
has a universal character (at {\it every} accessible world, $A$ is true).
We conclude $\Box A$ at world $w$ if we can show
that for all worlds $w'$ that are
accessible from $w$, $A$ is provable at $w'$; this is reflected in the
$\Box I$ rule.

The universal character of modal necessity would suggest
that we can {\it use} a proof of $\Box A [ w ]$ by exhibiting a 
world $w'$ accessible from $w$ and then assuming that $A$ was provable
there. 
\[
\infer[\Box{E}']
{\Gamma \ndcpl C[w]}
{\Gamma \ndcpl \Box A [ w ]  & \quad w \prec w' & \quad 
\Gamma \ndcpl A[w'] \longrightarrow \Gamma  \ndcpl C[w]}
\]
\noindent
Surprisingly, this rule is locally sound but 
not locally complete in the presence of
potentially infinite accessibility relations (consider an infinitely
branching accessibility relation -- this would require infinite
applications of $\Box E'$ in order to obtain enough to information to
re-apply $\Box I$), so {\bf CPL} uses
a less intuitive third-order formulation of $\Box E$ shown in
Fig.~\ref{fig:teth_nd}. The more intuitive rule is nevertheless derivable
from the actual $\Box E$ rule, and the third-order formulation of the 
rule is derivable
from $\Box E'$ under the assumption that we can finitely enumerate the
worlds accessible from any world (this is established in the file
\verb|AltBoxE.agda| in the Agda development). 

As per usual
in our development, we show our
rules to be locally sound, as witnessed by the following reduction:
\[
\begin{array}{c}
\infer[{\Box}E]
{\Gamma \ndArr C[w]\mathstrut}
{\infer[{\Box}I]{\Gamma \ndcpl \Box A[w]\mathstrut}
{\deduce{\forall w'.w\prec w' \longrightarrow \Gamma \ndcpl
    A[w']\mathstrut}{\mathcal D_1 \mathstrut}}
 & \deduce{(\forall w'. w \prec w' \longrightarrow \Gamma \ndcpl A[w'])
   \longrightarrow \Gamma \ndcpl C[w]}{\mathcal D_2 \mathstrut}\mathstrut}
\end{array}
\qquad
\]
\vspace{-10pt}
\[
\qquad\qquad\qquad\qquad\qquad\qquad\qquad\qquad\qquad\qquad\qquad\qquad\qquad
\qquad
\Longrightarrow_R
\quad
\begin{array}{c}
\deduce{\Gamma \ndcpl C[w] \mathstrut}{\mathcal D_2\,\mathcal D_1\mathstrut}
\end{array}
\]
Local completeness for modal necessity is the same as it
was for modal possibility; the second premise of the $\Box E$
rule essentially restates
the $\Box I$ rule.

Having shown our system to be locally sound and complete, we 
must now circle back around to show that the judgmental principles hold:
\begin{theorem}[Metatheory of {\bf CPL} natural deduction]
\label{thrm:meta_teth}~
\begin{itemize}
\item {\it Hypothesis principle}: If $A[w] \in \Gamma$, then $\Gamma \ndcpl A[w]$.
\item {\it Generalized weakening principle}: If $\Gamma \subseteq_w \Gamma'$ and $\Gamma \ndcpl A[w]$, then $\Gamma' \ndcpl A[w]$.
\item {\it Substitution principle}: If $\Gamma \ndcpl A[w]$ and 
 $\Gamma, A[w] \ndcpl C[w]$, then $\Gamma \ndArr C[w]$.
\end{itemize}
\end{theorem}
\begin{proof}
The hypothesis principle follows immediately from the rule {\it hyp}. 
The generalized weakening principle is established by structural induction
on given derivation, and the substitution principle is established by
structural induction on the second given derivation $\Gamma, A[w] \ndcpl C[w]$.
Both proofs appear in \verb|TetheredCPL/NatDeduction.agda| 
in the Agda development.
\end{proof}

\subsection{Sequent calculus}

Often we want to be able to
show that a judgment is not provable in a logic (for instance, we
better not be able to derive the judgment $\cdot \ndcpl \bot[w]$, which
would represent a closed contradiction). While natural
deduction is a canonical way of thinking about proofs, it is not
very useful as a tool for proving such negative statements about
logic. This is largely because 
natural deduction does not obey the so-called
\emph{sub-formula property} (all judgments in a proof refer only to
sub-formulas of the propositions present in the initial judgment). 
A sequent calculus system, on the other hand, 
obeys the sub-formula property and therefore
allows us to prove negative statements about a logic by refutation:
we assume the sequent is provable and, by case analysis on the 
structure of the derivation, derive a contradiction. The sequent
calculus 
for {\bf CPL} is given in Fig.~\ref{fig:teth_seq}.

\begin{figure}[tp]
\[
\infer[{\it init} \mbox{~~~($Q$ is an atomic proposition)}]{\Gamma,
  Q[w]  \seqcpl Q[w] }{}
 \quad
\infer[\bot{L}]{\Gamma \seqcpl C[w]}{\bot[w] \in \Gamma}
\]
\[
\infer[{\supset}R]{\Gamma \seqcpl A\supset B[w] }{\Gamma, A[w]
  \seqcpl B[w]} \quad 
\infer[{\supset}L]{\Gamma \seqcpl C[w] }{A\supset B[w] \in \Gamma
  & \Gamma \seqcpl A[w]  & \Gamma, B[w] \seqcpl C[w]} 
\]
\[
\infer[{\Diamond}R]{\Gamma \seqcpl \Diamond A[w]}{w\prec w' \quad \Gamma
\seqcpl A[w']} \quad
\infer[{\Box}R]{\Gamma \seqcpl \Box A [w]}{\forall w'.\, w\prec w'
  \longrightarrow \Gamma \seqcpl A[w']  } 
\]
\[
\infer[{\Diamond}L]{\Gamma \seqcpl C[w]}
{\Diamond A[w] \in \Gamma & \forall
  w'.\, w\prec w' \longrightarrow \Gamma \seqcpl A[w'] \longrightarrow
  \Gamma  \seqcpl C[w]}
\]
\[
\infer[{\Box}{L}]{\Gamma\seqcpl C[w]}{ \Box A[w] \in \Gamma & (\forall
  w'.\, w\prec w' \longrightarrow \Gamma\seqcpl A[w']) \longrightarrow 
 \Gamma \seqcpl C[w]}
\]
\caption{Sequent calculus for intuitionistic {\bf CPL}}
\label{fig:teth_seq}
\end{figure}

Even though sequent calculus systems are structured quite differently
than natural deduction systems, we can (and must!) establish the 
admissibility of the same defining principles.
\begin{theorem}[Metatheory of the {\bf CPL} sequent calculus]\label{thrm:meta_teth_seq}~
\begin{itemize}
\item Hypothesis principle: If $A[w_i] \in \Gamma$, then $\Gamma \seqcpl A[w_i]$.
\item Generalized weakening principle: If $\Gamma \subseteq_w \Gamma'$ and $\Gamma
  \seqcpl A[w]$, then $\Gamma' \seqcpl A[w]$.
\item Substitution principle: If $\Gamma \seqcpl A[w]$ and $\Gamma , A [ w ] \seqcpl C [ w
]$, then $\Gamma \seqcpl C [ w ]$.
\end{itemize}
\end{theorem}

\begin{proof}
The hypothesis principle is established by structural induction on the 
proposition $A$, and the generalized weakening principle is established
by structural induction on the given derivation. The substitution principle
is proved by lexicographic induction, primarily on the structure of the
proposition $A$ and secondarily on the structures of both given derivations:
if the proposition $A$ stays the same, then either the first derivation
gets smaller and the second stays the same or the second derivation gets
smaller and the first stays the same.
All proofs appear in \verb|TetheredCPL/Sequent.agda| in the Agda development.
\end{proof}
In sequent calculi, the hypothesis principle is frequently called
{\it identity admissibility} and the substitution principle is frequently
called {\it cut admissibility}. The admissibility of cut and identity 
establish the global analogues of local soundness and
completeness, respectively. 

By presenting a sequent calculus system as a convenient way of establishing
non-provability of hypothetical judgments in a natural deduction system, 
we have presupposed that the two presentations are equivalent. Luckily,
we were right:

\begin{theorem}[Equivalence]
$\Gamma \ndcpl A [w]$ if and only if $\Gamma \seqcpl A [w]$. 
\end{theorem}

\begin{proof}
Both directions must be proved simultaneously, primarily by 
induction on the accessibility relation and secondarily by 
structural induction on the given derivation. The defining principles
of the sequent calculus presentation (Theorem~\ref{thrm:meta_teth_seq}) 
are used in 
the forward direction, and the defining principles of the natural deduction 
presentation (Theorem~\ref{thrm:meta_teth}) are used in the backward direction.
The proof appears in \verb|TetheredCPL/Equiv.agda| in the Agda development.
\end{proof}

\subsection{Example}

We now formalize the example that motivated
Section~\ref{sec:reflect2}, showing that the sequent
$\Diamond Q[\alpha] \seqcpl \bot[\alpha]$ is derivable (and by the
equivalence of natural deduction and sequent calculus, that $\Diamond
Q[\alpha] \ndcpl \bot[\alpha]$ is derivable). The last rule in our 
proof will be $\Diamond L$:
\[
\infer[\Diamond L]{\Diamond Q[\alpha] \seqcpl \bot [\alpha]}
{\deduce{\forall w'.\, \alpha \prec w' 
\longrightarrow \Diamond Q[\alpha] \seqcpl Q [w'] 
\longrightarrow \Diamond Q[\alpha] \seqcpl \bot [\alpha]}{\mathcal E} }
\]
Therefore, it suffices to show that for all $w'$ accessible from $\alpha$,
$\Diamond Q[\alpha] \seqcpl Q[w']$ implies
$\Diamond Q[\alpha] \seqcpl \bot[\alpha]$. In this running
example, there are two worlds $\beta$ and $\gamma$ accessible from $\alpha$, 
so we must show that $\Diamond Q[\alpha] \seqcpl Q[\beta]$ implies 
$\Diamond Q[\alpha] \seqcpl \bot[\alpha]$ and that
$\Diamond Q[\alpha] \seqcpl Q[\gamma]$ implies 
$\Diamond Q[\alpha] \seqcpl \bot[\alpha]$. The reasoning in both cases
is exactly the same; we'll prove only the first here.

The way we prove that $\Diamond Q[\alpha] \seqcpl Q[\beta]$ implies 
$\Diamond Q[\alpha] \seqcpl \bot[\alpha]$ is to prove that there is 
{\it no} proof of $\Diamond Q[\alpha] \seqcpl Q[\beta]$, which means
that the implication holds vacuously. To prove this, we assume $\Diamond
Q[\alpha] \seqcpl Q[\beta]$ is derivable. The only possible rule that
could potentially allow us to conclude this sequent is $\Diamond L$, since there
is no $Q[\beta]$ in the context in order to apply the {\it init}
rule. However, since the worlds $\alpha$ and $\beta$ do not match, the
rule does not apply and the sequent is not provable.

\section{{\bf CPL*}, De-tethered constructive provability
  logic}\label{sec:cpl*}

\begin{figure}[tp]
\[
\infer[{\it hyp}]{\Gamma,
  A[w]  \ndcpls A[w] }{}
 \quad
\infer[\bot{E}]{\Gamma \ndcpls C[w']}{w' \prec^* w & \Gamma \ndcpls \bot[w] }
\]
\[
\infer[{\supset}I]{\Gamma \ndcpls A\supset B[w] }{\Gamma, A[w]
  \ndcpls B[w]} \quad 
\infer[{\supset}E]{\Gamma \ndcpls B[w] }{\Gamma \ndcpls A\supset B[w]
  & \Gamma \ndcpls A[w]} 
\]
\[
\infer[{\Diamond}I]{\Gamma \ndcpls \Diamond A[w]}{w\prec w' \quad \Gamma
\ndcpls A[w']} \quad
\infer[{\Box}I]{\Gamma \ndcpls \Box A [w]}{\forall w'.\, w\prec w'
  \longrightarrow \Gamma \ndcpls A[w']  } 
\]
\[
\infer[{\Diamond}E]{\Gamma \ndcpls C[w'']}
{w'' \prec^* w & \Gamma \ndcpls \Diamond A[w] & \forall
  w'.\, w\prec w' \longrightarrow \Gamma \ndcpls A[w'] \longrightarrow
  \Gamma  \ndcpls C[w'']}
\]
\[
\infer[{\Box}{E}]{\Gamma\ndcpls C[w'']}{ w'' \prec^* w  & \Gamma \ndcpls \Box A[w]  & (\forall
  w'.\, w\prec w' \longrightarrow \Gamma\ndcpls A[w']) \longrightarrow 
 \Gamma \ndcpls C[w'']}
\]
\caption{Intuitionistic {\bf CPL*} natural deduction}
\label{fig:deteth_nd}
\end{figure}

The natural deduction rules for {\bf CPL*}
are presented in Figure~\ref{fig:deteth_nd}. The only
difference from the corresponding rules of the previous section is that
we no longer restrict the conclusion of elimination rules to be at the world
$w$ of the judgment we are eliminating, instead allowing it to be at a world 
$w''$, 
provided that $w'' \prec^* w$ (an exception is ${\supset} E$,
since the rule does not mention an arbitrary proposition $C$).

The proofs of local soundness and completeness are analogous to the ones
discussed in the previous section; the substitution principle
is de-tethered in the same way that the elimination rules are.
\begin{theorem}[Metatheory of {\bf CPL*} natural deduction]
\label{thrm:meta_deteth}~
\begin{itemize}
\item {\it Hypothesis principle}: If $A[w] \in \Gamma$, then $\Gamma \ndcpls A[w]$.
\item {\it Generalized weakening principle}: If $\Gamma \subseteq_w \Gamma'$ and $\Gamma \ndcpls A[w]$, then $\Gamma' \ndcpls A[w]$.
\item {\it Substitution principle}: If $w' \prec^* w$, 
$\Gamma \ndcpls A[w]$, and 
$\Gamma, A[w] \ndcpls C[w']$, then 
$\Gamma \ndcpls C[w']$.
\end{itemize}
\end{theorem}
\begin{proof}
The hypothesis principle again follows immediately from the rule {\it hyp}. 
The generalized weakening principle is established by a primary
induction on the accessibility relation and a secondary
structural induction
on the given derivation. The substitution principle is established
by a primary induction on the accessibility relation and then a 
secondary structural induction on the second given derivation 
$\Gamma, A[w] \ndcpls C[w']$.
Both proofs appear in \verb|DetetheredCPL/NatDeduction.agda| 
in the Agda development.
\end{proof}

\subsection{Focused sequent calculus}

The sequent calculus formulation of {\bf CPL} is convenient for 
establishing very simple properties of provability and non-provability,
and it is possible to give a very similar sequent calculus 
for {\bf CPL*} \cite{simmons11constructive}.
However, because we wish to consider {\bf CPL*} as the basis of 
a logic programming language, we follow Andreoli \shortcite{andreoli92logic}
in developing a much more restricted {\it focused} sequent calculus. 
Unlike Andreoli, we use an explicitly polarized version of our logic.

Propositions in a polarized presentation of logic are split into 
two syntactic categories, {\it positive} 
propositions $A^+$ and {\it negative} propositions $A^-$. 
A full discussion of polarity assignment for connectives
is outside the scope of this
article; as a rule of thumb, the {\it positive} connectives are those
with large eliminations. An elimination is large when the proposition whose
truth is established by an elimination rule 
is some proposition $C$ with no immediate connection to the
proposition being eliminated; this indicates that
$\bot$, $\Diamond A$, and $\Box A$ are positive connectives and $A \supset B$
is not.
\[
\begin{array}{lcl}
A^+,B^+ & ::= & Q^+ \mid \;{\downarrow} A^- \mid \bot \mid \Diamond A^+\mid \Box
A^+\\
A^-,B^- & ::= & Q^- \mid \;{\uparrow} A^+ \mid A^+ \supset B^-
\end{array}
\]
Each atomic 
proposition can be positive or negative, but never both, as if each
atomic proposition in the un-polarized logic was always already intrinsically
positive or negative and our previous natural deduction and sequent
calculi were unable to notice.

Both of the modal operators in constructive provability logic
are naturally positive on the {\it outside}. 
However, our choice of the polarity for
the proposition {\it inside} the modality appears to be arbitrary: 
$\Diamond A^+$
and $\Diamond A^-$ would both be reasonable ways to polarize
the possibility modality. Polarization of propositions
is a property that affects {\it proofs},
not {\it provability}, so the modalities $\Diamond A$ and $\Box A$, 
which, in constructive provability logic, 
only care about the provability of the sub-formula $A$,
are naturally
indifferent to the treatment of $A$ as a positive or negative proposition.

\begin{figure}
\fbox{$\rfoc{\Gamma}{C^-[w]}$}
\[
\infer[QR^+]{\rfoc{\Gamma , Q^+[w]}{Q^+[w]}}{}
\qquad
\infer[{\downarrow} R]{\rfoc{\Gamma}{{\downarrow} A^-[w]}}
{\inv{\Gamma ; \cdot}{A^-[w]}}
\]
\[
\infer[\Diamond R]{\rfoc{\Gamma}{\Diamond A^+ [w]}}
{w \prec w' & \inv{\Gamma ; \cdot}{{\uparrow}A^+[w']} }
\qquad
\infer[\Box R]
{\rfoc{\Gamma}{\Box A^+[w]}}
{\forall w'. w \prec w' \longrightarrow \inv{\Gamma ; \cdot}{{\uparrow}A^+[w']}}
\]

\fbox{$\inv{\Gamma; \cdot}{C^-[w]}$}
\[
\infer{Q^+\,\mathit{stable}^+\mathstrut}{}
\qquad
\infer{{\downarrow}A^-\,\mathit{stable}^+\mathstrut}{}
\qquad
\infer{Q^-\,\mathit{stable}^-\mathstrut}{}
\qquad
\infer{{\uparrow}A^+\,\mathit{stable}^-\mathstrut}{}
\]
\[
\infer[{\supset} R]{\inv{\Gamma ; \cdot}{A^+ \supset B^-[w]}}
{\inv{\Gamma ; A^+[w]}{B^-[w]}}
\qquad
\infer[L]{\inv{\Gamma ; A^-[w']}{C^-[w]}}
{A^+\;\stable^+ & \inv{\Gamma , A^+[w'] ; \cdot}{C^-[w]}}
\]
\[
\infer[{\downarrow} L]
{\inv{\Gamma , {\downarrow} A^-[w'] ; \cdot}{C^-[w]}}
{C^-\,\stable^- &
 w \prec^* w' & \lfoc{\Gamma, {\downarrow} A^-[w']}{A^-[w']}{C^-[w]}}
\]
\[
\infer[\bot L]{\inv{\Gamma ; \bot[w']}{C^-[w]}}{}
\qquad
\infer[\Diamond L]{\inv{\Gamma ; \Diamond A^+[w']}{C^-[w'']}}
{\forall w. w' \prec w \longrightarrow \inv{\Gamma ; \cdot }{{\uparrow}A^+[w]}
  \longrightarrow \inv{\Gamma ; \cdot}{C^-[w'']} }
\]
\[
\infer[\Box L]{\inv{\Gamma ; \Box A^+[w']}{C^-[w'']}}
{(\forall w. w' \prec w \longrightarrow \inv{\Gamma ; \cdot}{{\uparrow}A^+[w]})
  \longrightarrow \inv{\Gamma ; \cdot}{C^-[w'']}}
\qquad
\infer[{\uparrow} R]{\inv{\Gamma ; \cdot}{{\uparrow} A^+[w]}}
{\rfoc{\Gamma}{A^+[w]}}
\]

\fbox{$\lfoc{\Gamma}{A^-[w']}{C^-[w]}$}
\[
\infer[QL^-]{\lfoc{\Gamma}{Q^-[w]}{Q^-[w]}}{}
\qquad
\infer[{\uparrow} L]{\lfoc{\Gamma}{{\uparrow} A^+[w']}{C^-[w]}}
{\inv{\Gamma ; A^+[w']}{C^-[w]}}
\]
\[
\infer[{\supset} L]{\lfoc{\Gamma}{A^+ \supset B^-[w']}{C^-[w]}}
{\rfoc{\Gamma}{A^+[w']} & \lfoc{\Gamma}{B^-[w']}{C^-[w]} }
\]
\caption{Focused sequent calculus for intuitionistic {\bf CPL*}}\label{fig:foc_seq}
\end{figure}

To develop the focused calculus, we require three types
of sequent: a \emph{right focus} sequent $\rfoc{\Gamma}{A^+[w]}$, describing a
state where non-invertible right rules are applied to positive
propositions; a \emph{left focus} sequent $\lfoc{\Gamma}{A^-[w']}{C^-[w]}$,
where non-invertible left rules are applied to negative propositions
(we typically say that the proposition $A^-$ is under focus);
and an \emph{inversion} sequent $\inv{\Gamma ; \Omega}{A^-[w]}$,
describing everything else (the additional context $\Omega$, which 
is either $\cdot$ or a single judgment $A^+[w]$, 
is called the \emph{inversion context}). We define the
system in such a way that whenever the inversion context is non-empty,
there is only one applicable rule -- the one that decomposes the
connective in the inversion context. We require two additional
judgments, $A^+\,\stable^+$ and $A^-\,\stable^-$, which
restrict the inversion phase. 
The rules defining the focused {\bf CPL*} sequent
calculus are given in Fig.~\ref{fig:foc_seq}.

\begin{figure}
\begin{align*}
 & 
 & &
 & (\cdot)^\circledcirc & = \cdot
\\
(Q^+)^\oplus & = Q^+
 & (Q^+)^\ominus & = {\uparrow}Q^+
 & (\Gamma, Q^+[w])^\circledcirc & = \Gamma^\circledcirc, Q^+[w]
\\
(\bot)^\oplus & = \bot
 & (\bot)^\ominus & = {\uparrow}\bot
 & (\Gamma, \bot[w])^\circledcirc & = \Gamma^\circledcirc, {\downarrow} {\uparrow} \bot[w]
\\
(\Diamond A)^\oplus & = \Diamond A^\oplus 
 & (\Diamond A)^\ominus & = {\uparrow}(\Diamond A^\oplus)
 & (\Gamma, \Diamond A[w])^\circledcirc & = \Gamma^\circledcirc, {\downarrow}{\uparrow}(\Diamond A^\oplus)[w]
\\
(\Box A)^\oplus & = \Box A^\oplus 
 & (\Box A)^\ominus & = {\uparrow}(\Box A^\oplus)
 & (\Gamma, \Box A[w])^\circledcirc & = \Gamma^\circledcirc, {\downarrow}{\uparrow}(\Box A^\oplus)[w]
\\
(Q^-)^\oplus & = {\downarrow}Q^-
 & (Q^-)^\ominus & = Q^-
 & (\Gamma, Q^-[w])^\circledcirc & = \Gamma^\circledcirc, {\downarrow}Q^-[w]
\\
(A \supset B)^\oplus & = {\downarrow}(A^\oplus \supset B^\ominus)
 & (A \supset B)^\ominus & = A^\oplus \supset B^\ominus
 & (\Gamma, A \supset B[w])^\circledcirc 
 & = \Gamma^\circledcirc, {\downarrow}(A^\oplus \supset B^\ominus)[w]
\end{align*}
\caption{Polarization of propositions and contexts}\label{fig:polarize}
\end{figure}

Validating the judgmental principles is quite complex in focused {\bf CPL*}; 
the proof adapts techniques used in the analogous proofs for {\bf CPL}
as well as the {\it structural focalization} techniques described
by Simmons \shortcite{simmons11structural}.
The generalized weakening principle is established in 
\verb|FocusedCPL/Weakening.agda|.
The substitution principle is established as a corollary of the
cut admissibility, which is established in 
\verb|FocusedCPL/Cut.agda|. Notably, in order to prove the substitution 
theorem, we must simultaneously prove a 
a {\it backwards} substitution theorem establishing that
$\inv{\Gamma;\cdot}{A[w]}$ and $\inv{\Gamma;\cdot}{C[w']}$ together
imply $\inv{\Gamma;\cdot}{C[w']}$; this fact does not follow
from generalized weakening when $w' \prec^+ w$. 
Finally, the hypothesis principle is established as a corollary of 
identity expansion in \verb|FocusedCPL/Identity.agda|.

We only establish a weak form of equivalence between the focused sequent
calculus and the natural deduction system; we define a polarization
strategy (Figure~\ref{fig:polarize}) that maps unpolarized propositions and
contexts to polarized ones. It is more robust to define equivalence
on the basis of {\it erasing} polarized propositions and contexts to
unpolarized ones \cite{simmons11structural},
but this formulation is sufficient for our purposes.

\begin{theorem}[Equivalence]
$\Gamma \ndcpls A[w]$ if and only if 
$\inv{\Gamma^\circledcirc; \cdot}{A^\ominus[w]}$
\end{theorem}

\begin{proof}
Both directions must be proved simultaneously, primarily by induction on
the accessibility relation and secondarily by structural induction on the
given derivation. 
The forward direction uses the metatheory of the focused sequent calculus
and is structured similarly to the proof 
in \cite{simmons11structural},
and the reverse direction uses the defining principles of the
natural deduction system (Theorem~\ref{thrm:meta_deteth}).
The proof appears in \verb|DetetheredCPL/Equiv.agda| in the Agda development.
\end{proof}

\section{Logic programming in constructive provability logic}\label{sec:logprog}

Proving the natural deduction system for {\bf CPL*} equivalent to a focused 
presentation of the logic is a lot of work, but the payoff is that the
focused sequent calculus can form the basis of a logic programming language
\cite{miller91uniform,andreoli92logic}. We will use an {\it extremely}
simplified example here: translating a propositional 
Horn
clause logic program with stratified negation where there are only two 
strata. In this section, atomic propositions in the first strata
will be written with the metavariable $Q$, and atomic propositions 
in the second strata will be written with the metavariable $P$. 

Atomic propositions $Q$ can appear at the head of Horn clauses 
of the form $Q \;\verb|:-|\; Q_1, \ldots, Q_n$ in the 
logic program;
atomic propositions $P$ can appear at the head of Horn clauses of 
the form
$P \;\verb|:-|\; A_1, \ldots, A_n$ in the logic program, 
where each $A_i$ is either 
an atomic proposition $P_i$, an atomic proposition $Q_i$, 
or a negated atomic proposition 
$\neg Q_i$. We will use the worlds $\beta$ and $\gamma$ (where 
$\beta \prec \gamma$) from our running example. Each first-strata 
Horn clause $Q \;\verb|:-|\; Q_1, \ldots, Q_n$ is translated into
a judgment 
${\downarrow}(Q_1 \supset \ldots \supset Q_n \supset {\uparrow}Q) [\gamma]$,
and each second-strata Horn clause 
$P \;\verb|:-|\; A_1, \ldots, A_n$ is translated into a judgment
${\downarrow}(A_1^\bullet \supset \ldots \supset A_n^\bullet \supset {\uparrow}P) [\beta]$,
where $(P_i)^\bullet = P_i$, $(Q_i)^\bullet = \Box Q_i$, and 
$(\neg Q_i)^\bullet = {\downarrow}((\Box Q_i) \supset {\uparrow}\bot)$. 
(Note that this implies a positive polarity for all atomic propositions.)
We name the context
obtained by translating our Horn
clause logic program $\Gamma$.

Searching for a proof of a proposition $P$ using bottom-up logic programming
can be characterized as a two phase proof search procedure for proofs of the term
$\inv{\Gamma, \Gamma';\cdot}{{\uparrow}P[\beta]}$, where we always maintain
the invariant that $\inv{\Gamma, \Gamma';\cdot}{{\uparrow}P[\beta]}$ is 
provable
if and only if $\inv{\Gamma;\cdot}{{\uparrow}P[\beta]}$ is provable. 

In the first phase, we only focus on hypotheses in
$\Gamma$ with the form
${\downarrow}(Q_1 \supset \ldots \supset Q_n \supset {\uparrow}Q)[\gamma]$.
Because focusing on such a proposition 
will succeed exactly when $Q_i[\gamma] \in \Gamma'$ for each of the
$Q_i$, it is always possible to determine the 
entire set of $Q_k$ that are {\it immediate consequences} of the rules
in $\Gamma$ and 
atomic propositions in $\Gamma'$. Given a sequent 
$\inv{\Gamma, \Gamma';\cdot}{{\uparrow}P[\beta]}$ that is true if and 
only if $\inv{\Gamma;\cdot}{{\uparrow}P[\beta]}$, we determine the
immediate (first-strata)
consequences $\Gamma_{\it imm}$ of $(\Gamma, \Gamma')$. By repeated
focusing steps, we can show
$\inv{\Gamma, (\Gamma' \cup \Gamma_{\it imm});\cdot}{{\uparrow}P[\beta]}$
implies $\inv{\Gamma, \Gamma';\cdot}{{\uparrow}P[\beta]}$, and we can
show the converse by the reverse substitution principle discussed
in the previous section. This in turns means that we have
a new sequent 
$\inv{\Gamma, (\Gamma' \cup \Gamma_{\it imm});\cdot}{{\uparrow}P[\beta]}$
which is true if and only if $\inv{\Gamma;\cdot}{{\uparrow}P[\beta]}$.
If $\Gamma' \not\supseteq \Gamma_{\it imm}$, we repeat the first phase.
Otherwise $\Gamma' \supseteq \Gamma_{\it imm}$, so
all the immediate consequences $Q$
of $(\Gamma,\Gamma')$ are already present in $\Gamma'$. In this case, we say we
have reached {\it saturation at $\gamma$} and continue to the second phase.

The second phase relies on the fact that, if all of the immediate consequences
$Q$ of $(\Gamma,\Gamma')$ are already present in $\Gamma'$, then
$\inv{\Gamma,\Gamma'; \cdot}{{\uparrow}Q[\gamma]}$ 
is provable if and only if $Q[\gamma] \in (\Gamma, \Gamma')$. This means
that we have an effective 
decision procedure for the provability of first-strata
propositions $Q$. Thus, the second phase proceeds the same as the first,
focusing instead on hypotheses in $\Gamma$ with the form 
${\downarrow}(A_1^\bullet \supset \ldots \supset A_n^\bullet \supset 
{\uparrow}P)[\beta]$. Focusing on such a rule will succeed exactly when:
\begin{itemize}
\item for each $A_i^\bullet = P_i$, $P_i[\beta] \in \Gamma'$,
\item for each $A_i^\bullet = \Box Q_i$, $Q_i[\gamma] \in \Gamma'$, and
\item for each $A_i^\bullet = {\downarrow}((\Box Q_i) \supset {\uparrow}\bot)$, $Q_i[\gamma] \not\in \Gamma'$.
\end{itemize}
Therefore, given that $(\Gamma, \Gamma')$ is saturated at $\gamma$,
we can also determine the entire set of second-strata
propositions that are immediate 
consequences of $(\Gamma, \Gamma')$. We proceed as before, and 
once we have reached saturation at $\beta$ as well, we can declare the
original sequent $\inv{\Gamma; \cdot}{{\uparrow}P[\beta]}$ provable
if and only if $P[\beta] \in \Gamma'$ for the final
saturated $\Gamma'$. 

\section{Axiomatic characterization}\label{sec:axiom}


In this section, we present a sound
Hilbert-style proof theory for {\bf CPL} and {\bf CPL*}. 
The desired interpretation of $\hilcpl A$ is that it implies
that, for all converse well-founded accessibility relations and 
contexts $\Gamma$, it is the case that 
$\Gamma \ndcpl A[w]$ (in {\bf CPL}). 
Similarly, the desired interpretation of $\hilcpls A$ 
is that, for all converse well-founded accessibility
relations and contexts $\Gamma$, it is the case that 
$\Gamma \ndcpls A[w]$ (in {\bf CPL*}). We will write
$\hilcplo A$ to indicate results that hold in both {\bf CPL}
and {\bf CPL*}. 

This section only considers {\it soundness} results
for Hilbert-style reasoning; we do not claim the converse, which would be a 
{\it completeness} result.
However, when
we claim that a particular formula is not an axiom of
{\bf CPL} or {\bf CPL*}, we always can demonstrate
a particular accessibility relation, world, and instance $A$ of the said
formula such that there is no proof of $\Gamma \ndcpl A[w]$ or
$\Gamma \ndcpls A[w]$. 
For instance, 
$Q [ \alpha ] \ndcpl (\neg \Diamond Q \supset \Box \neg Q)
[ \alpha ]$
is unprovable,\footnote{$\neg A$ is the usual 
intuitionistic negation $A \supset \bot$} 
so the classically true De Morgan axiom
$\neg \Diamond A \supset \Box \neg A$ does not hold in {\bf CPL}.
Some axioms, like $\Box A \supset \Box \Box A$, only hold in general
when the accessibility relation is transitive; these are indicated.

 Both proofs and counterexamples for 
 {\bf CPL} and {\bf CPL*} can be found
in {\tt TetheredCPL/Axioms.agda} and
in {\tt DetetheredCPL/Axioms.agda} (respectively) in the 
Agda development.

\subsection{Intuitionistic modal logic}

All of the axioms of intuitionistic propositional logic are true 
in both variants of constructive provability logic, as are the
fundamental rules and axioms of intuitionistic modal logic. It is less clear
what other axioms characterize intuitionistic modal logic;
some of the axioms of Simpson's {\bf IK} hold in neither Pfenning-Davies
{\bf S4} nor in constructive provability logic.
\begin{theorem}[Intuitionistic modal logic]~\label{thrm:imla}
\begin{tabbing}
$\quad\it(MP)$~~~ \= 
$\hilcpl A \supset B$ and $\hilcpl A$ imply $\hilcpl B$, and 
$\hilcpls A \supset B$ and $\hilcpls A$ imply $\hilcpls B$\\
$\quad\it(I)$ \> $\hilcplo A \supset A$ \` \\
$\quad\it(K)$ \> $\hilcplo A \supset B \supset A$ \`  \\
$\quad\it(S)$ \> $\hilcplo (A \supset B \supset C) \supset (A \supset B) \supset A \supset C $ \`\\
$\quad\it(\bot E)$ \> $\hilcplo \bot \supset A$ \` \\
$\quad\it(NEC)$ \> $\hilcpl A$ implies $\hilcpl \Box A$, and $\hilcpls A$ implies
$\hilcpls \Box A$\\
$\quad\it(K\Box)$~ \> $\hilcplo \Box (A \supset B) \supset \Box A \supset \Box B$ \`  \\
$\quad\it(K\Diamond)$~ \> $\hilcplo \Box (A \supset B) \supset \Diamond A \supset \Diamond B$ \`\\
$\quad\it(4\Box)$~ \> $\hilcplo \Box A \supset \Box \Box A$
 \` (if the accessibility relation is transitive)\\
$\quad\it(\Diamond\bot)$ ~\> $\hilcpls \neg \Diamond \bot$ \` \\
$\quad\it(4\Diamond)$~ \> $\hilcpls \Diamond \Diamond A \supset \Diamond A$ \` (if the accessibility relation is transitive)
\end{tabbing}
$\neg \Diamond \bot$ is not an axiom of {\bf CPL}, and 
$(\Diamond A \supset \Box B) \supset \Box (A \supset B)$ 
is not an axiom of either variant.
\end{theorem}

\noindent
If the accessibility relation is transitive, {\bf CPL*} admits  the
axioms of Pfenning-Davies {\bf S4}, plus $(\Diamond \bot)$, which 
holds in {\bf IK} but not in Pfenning-Davies {\bf S4}. 
We have not been able establish the status of axiom 
$\it 4\Diamond$ in {\bf CPL}.

Simpson's thesis presents axioms characterizing other properties of 
accessibility relations besides transitivity, but all these properties 
(e.g. symmetry) are inconsistent with converse well-foundedness, so 
we ignore them here. 

\subsection{Provability logic}

Exploring the connection between constructive provability logic and provability
logic was one of the motivations of this work. The most common 
characterization
of provability logic is the $\it GL$ axiom. Since $\it GL$ can be
used to prove the $\it 4\Box$ axiom \cite{verbrugge10provability}, it is not 
surprising that this axiom requires a transitive accessibility relation.
The other standard characterization of provability logic is the L\"ob rule.
The L\"ob rule is almost always presented together with axiom 
$\it 4 \Box$ ensuring transitivity of the accessibility relation, but it is 
interesting to observe that the L\"ob rule, unlike the $\it  GL$ axiom, holds 
even without a transitive accessibility relation.

\begin{theorem}[Provability logic] ~\label{thrm:gl}
\begin{tabbing}
$\quad{\it (GL)}$ ~ \= $\hilcplo \Box (\Box A \supset A) \supset \Box A$
\` (if the accessibility relation is transitive)
\\
$\quad{\it (\mbox{\it L{\"o}b})}$ \> $\hilcpl \Box A \supset A$ implies $\hilcpl A$, and
$\hilcpls \Box A \supset A$ implies $\hilcpls A$
\end{tabbing}
\end{theorem}

Unlike the proofs of Theorem~\ref{thrm:imla}, 
both parts of Theorem~\ref{thrm:gl} are proved by induction over the 
accessibility relation.

\subsection{De Morgan laws}

The interaction between negation and the modal operators is frequently
an interesting ground for exploration. In classical modal logic, 
$\Diamond A$ is just defined as $\neg \Box \neg A$, and so 
all four of the De Morgan laws --
($\Diamond\neg     A \supset \neg     \Box     A$), 
($\Box    \neg     A \supset \neg     \Diamond A$), 
($\neg    \Diamond A \supset \Box     \neg     A$), and 
($\neg    \Box     A \supset \Diamond \neg     A$) -- hold trivially.
The first three hold in Simpson's {\bf IK}, and none hold in
Pfenning-Davies {\bf S4}. In {\bf CPL*} two of the four hold, 
and in {\bf CPL} the same two hold only if we make certain assumptions
about consistency at accessible worlds.

\begin{theorem}[De Morgan laws]~
\begin{itemize}
\item In {\bf CPL*}, 
$\hilcpls \Diamond\neg     A \supset \neg     \Box     A$ and
$\hilcpls \Box    \neg     A \supset \neg     \Diamond A$.

\item In {\bf CPL}, neither 
$\Diamond\neg     A \supset \neg     \Box     A$ nor
$\Box    \neg     A \supset \neg     \Diamond A$ are axioms.

\item In {\bf CPL}, both 
$\Gamma \Rightarrow \Diamond\neg     A \supset \neg     \Box     A [ w ]$ and
$\Gamma \Rightarrow \Box    \neg     A \supset \neg     \Diamond A [ w ]$
are true if there is no $w \prec w'$ such that 
$\Gamma \Rightarrow \bot [ w' ]$.

\item $\neg    \Diamond A \supset \Box     \neg     A$ is not an axiom 
of {\bf CPL} or {\bf CPL*}.

\item $\neg    \Box     A \supset \Diamond \neg     A$ is not an axiom
of {\bf CPL} or {\bf CPL*}.

\end{itemize}
\end{theorem}


\section{Conclusion}\label{sec:conc}

In this article, we have 
given natural deduction and sequent calculus presentations for
two variants of constructive provability
logic, a modal logic with reflection over both accessibility and
provability. The standard judgmental principles of all 
four deductive systems were presented and formalized in the Agda
proof assistant (with some caveats described in Section~\ref{sec:note}). 
Furthermore, through a focused sequent calculus
presentation, we produced a sketch of how constructive provability logic 
can be used as a intuitionistic and proof-theoretic justification for 
stratified negation 
in logic programming. Finally, as customary in
most works on provability logic, we gave a axiomatic
characterization of constructive provability logic and 
showed that most of the standard axioms of provability logic are sound 
with respect to our proof theoretic presentation.

\subsection{Related work}

There has been a substantial amount of research on provability logic
throughout the years. The early research on the topic focused on
axiomatic presentations of provability logic and its implications for
the foundations of mathematics. More recently, there has been interest
in the proof theoretic aspects of provability logic, mostly following
the cut elimination result of Valentini 
\shortcite{valentini83modal}.
However, most research in provability logic focuses on classical
logic (a detailed survey is given in
\cite{artemov04provability}). 
Intuitionistic formulations of provability logic have
historically been much less explored, with some notable
exceptions. For a more detailed historical account of intuitionistic provability
logic, as well as a development of a provability logic for intuitionistic
arithmetic, see \cite{iemhoff01provability}. 

Our line of work departs substantially from previous presentations,
even from intuitionistic variants of provability logic. 
Natural deduction systems for provability
logic are also not very common in the literature, given the historical
bias towards axiomatic systems. 
Furthermore, most existing
sequent calculi for provability logic are classical, and do not make
use of explicit worlds nor reflection, which arise as a natural way
of representing provability logic through the judgmental methodology, 
and thus are substantially different from our own. A focused sequent
calculus for provability logic is also, to the best of our knowledge,
unheard of. 

\subsection{Future work}

This work introduces propositional constructive logic programming as 
a modal logic. The only major shortcoming to our treatment of {\bf CPL}
and {\bf CPL*} as modal logics is that we do not know how to formulate
or prove the completeness of our system with respect to a Hilbert-style 
presentation. It is not at all clear
how this deficiency can be overcome. It may require the introduction of
a notion of validity similar to the validity considered
by Pfenning and Davies \shortcite{pfenning01judgmental}, and it also may require
more fundamental changes to the logic, such as making the computational
content of the higher order rule formulations more explicit.

In contrast to our relatively thorough investigation of {\bf CPL} and
{\bf CPL*} as modal logics, we have only barely
scratched the surface of understanding
the possible applications of constructive provability logic as the
basis for proof search and logic programming. We ultimately wish to use
constructive provability logic to justify the L10 logic programming
language, a rich forward-chaining language that uses worlds to enable
both distributed logic programming and locally stratified negation
\cite{simmons11distributed}. To do so, we require a
satisfactory treatment of first-order quantification in constructive 
provability logic, as the account in this paper was entirely propositional. 
In addition, it is likely that a hybrid modal operator $A @ w$ will 
prove to be more useful than the traditional modal operators
$\Diamond A$ and $\Box A$, but this
is a minor change from a proof-theoretic perspective. 

Horn-clause logic programming is only the simplest logic programming 
application of constructive provability logic; the focused presentation of
{\bf CPL*} immediately opens the door to the principled addition of
stratified negation to more interesting logic programming
languages, such as higher-order logic programming languages like
$\lambda$Prolog and Twelf. We also believe that constructive
provability logic with nominal quantification could be presented 
as a generalization of the Bedwyr language, which synthesizes 
model checking and logic programming \cite{baelde07bedwyr}.

Finally, provability logic is 
quite important in other areas of computer science,
particularly as the basis for the approximation or delay 
modality $\rhd$ used to model programming languages 
\cite{nakano00modality,richards10approximation}. 
We hope to better understand whether and how constructive provability logic
can relate to this line of work.

\subsection*{Acknowledgements}

Michael Ashley-Rollman, William Lovas, Frank Pfenning, Andr{\'e} Platzer,
and the reviewers and participants of the 2011 IMLA workshop provided
valuable feedback and corrections to earlier versions and drafts of
this work.

This work was supported by an X10 Innovation Award from IBM, a 
National Science Foundation Graduate Research Fellowship for the first author, 
and by Funda\c{c}\~{a}o para a Ci\^{e}ncia e a Tecnologia
(Portuguese Foundation for Science and Technology)
through the Carnegie Mellon Portugal Program under Grants NGN-44
and SFRH / BD / 33763 / 2009.

\bibliographystyle{acmtrans}
\bibliography{ref}

\end{document}